\documentclass[a4paper,11pt]{article}
\usepackage{amsmath,amssymb,amsthm,mathtools}
\usepackage{microtype,aliascnt}
\usepackage{enumitem,booktabs,array,needspace,placeins}
\usepackage{tikz}
\usepackage{authblk}
\usepackage{cite}
\usetikzlibrary{arrows.meta,positioning,fit,calc,backgrounds}
\definecolor{UCTblue}{HTML}{2864A0}
\definecolor{UCTteal}{HTML}{087F8C}
\definecolor{UCTamber}{HTML}{B97912}
\definecolor{UCTpurple}{HTML}{7554A3}
\definecolor{UCTred}{HTML}{B4434F}
\definecolor{UCTgray}{HTML}{65717D}
\definecolor{LFourBlue}{HTML}{28669B}
\definecolor{LFourOrange}{HTML}{AD650F}
\definecolor{LFourInk}{HTML}{263342}
\definecolor{LFourGray}{HTML}{545C66}
\tikzset{
  >=Stealth,
  uctvertex/.style={circle,draw,thick,minimum size=7mm,inner sep=1pt,font=\small},
  uctcenter/.style={uctvertex,draw=UCTblue,fill=UCTblue!9,text=UCTblue},
  uctchoice/.style={uctvertex,draw=UCTteal,fill=UCTteal!9,text=UCTteal},
  uctbox/.style={draw,thick,rounded corners=1mm,inner sep=2mm,align=center,font=\small},
  uctcore/.style={uctbox,draw=UCTblue,fill=UCTblue!5,text=UCTblue},
  uctcode/.style={uctbox,draw=UCTpurple,fill=UCTpurple!7,text=UCTpurple},
  uctjoin/.style={draw=UCTgray,line width=1pt},
  uctarrow/.style={->,draw=UCTgray,thick},
  uctlabel/.style={font=\small,align=center}
}
\usepackage{geometry}
\usepackage[hidelinks]{hyperref}
\usepackage{cleveref}
\hypersetup{pdftitle={W[1]-Hardness of Upper Clique Transversal: A Proof by Direct Construction}}
\theoremstyle{definition}
\newtheorem{theorem}{Theorem}
\newaliascnt{lemma}{theorem}
\newtheorem{lemma}[lemma]{Lemma}
\aliascntresetthe{lemma}
\crefname{theorem}{theorem}{theorems}
\Crefname{theorem}{Theorem}{Theorems}
\crefname{lemma}{lemma}{lemmas}
\Crefname{lemma}{Lemma}{Lemmas}
\crefname{figure}{figure}{figures}
\Crefname{figure}{Figure}{Figures}
\crefname{table}{table}{tables}
\Crefname{table}{Table}{Tables}
\crefname{equation}{equation}{equations}
\Crefname{equation}{Equation}{Equations}
\crefname{section}{section}{sections}
\Crefname{section}{Section}{Sections}
\renewcommand{\proofname}{Proof}
\makeatletter
\renewenvironment{proof}[1][\proofname]{\par
  \pushQED{\qed}%
  \normalfont\topsep6\p@\@plus6\p@\relax
  \trivlist
  \item[\hskip\labelsep\itshape\textcolor{black}{#1}\@addpunct{.}]\ignorespaces
}{%
  \popQED\endtrivlist\@endpefalse
}
\makeatother
\newcommand{\UCT}{\textnormal{\textsc{Upper Clique Transversal}}}
\newcommand{\MCC}{\textnormal{\textsc{Multicolored Clique}}}
\newcommand{\mcK}{\mathcal K}

\setlist{topsep=3pt,itemsep=2pt,parsep=0pt}
\title{W[1]-Hardness of Upper Clique Transversal\thanks{This work is partially supported by the Slovenian Research and Innovation Agency (I0-0035, research program P1-0285 and research projects J1-60012, J1-70035, J1-70046, and N1-0370), and by
JSPS KAKENHI Grant Numbers 
JP21K11757, 
JP22H00513, 
JP25K03077, 
and JST, CRONOS, Japan Grant Number JPMJCS24K2.}}

\author[1]{J.~Pascal Gollin}
\author[2]{Tesshu Hanaka}
\author[3]{Ekkehard K{\"o}hler}
\author[1,4]{Martin Milani{\v{c}}}
\author[5]{Yushi Uno}

\affil[1]{FAMNIT, University of Primorska, Slovenia}
\affil[2]{Kyushu University, Japan}
\affil[3]{Brandenburg University of Technology, Germany}
\affil[4]{IAM, University of Primorska, Slovenia}
\affil[5]{Osaka Metropolitan University, Japan}
\date{}

\date{}

\begin{document}

\maketitle

\thispagestyle{plain}

\begin{abstract}
A clique transversal of a graph is a set of vertices intersecting every maximal clique.
We prove that deciding whether a graph has an inclusion-wise minimal clique transversal of size at least $k$ is \textsf{W[1]}-hard when parameterized by $k$. 
\end{abstract}

\section{Introduction}

All graphs considered in this note are finite, simple, and undirected.
For a graph $G$, let $\mcK(G)$ be the set of all maximal cliques in $G$. 
A \emph{clique transversal} is a set $S\subseteq V(G)$ that intersects every $K\in\mcK(G)$, and it is \emph{minimal} if it is inclusion-wise minimal. 
In {the context of} \UCT{}, given a graph $G$ and a positive integer $k$, we ask whether $G$ has a minimal clique transversal of size at least $k$~\cite{MilanicUno2023}. 
Milani\v{c} and Uno~\cite{MilanicUno2023} proved \textsf{NP}-completeness of the problem on chordal graphs and linear-time solvability on split graphs and proper interval graphs. 
Jaffke, de Lima, and Nikabadi~\cite{JaffkeEtAl2026} gave polynomial-time algorithms for the class of interval graphs and, more generally, rooted directed path graphs.

\UCT{} is one of many \emph{max-min} (or \emph{upper}) problems, which maximize an inclusion-wise minimal feasible solution; similarly, \emph{min-max} problems minimize an inclusion-wise maximal one. 
Examples of max-min problems include \textsc{Upper Domination}~\cite{BazganEtAl2018}, \textsc{Upper Total Domination}~\cite{BazganEtAl2018}, \textsc{Maximum Minimal Vertex Cover}~\cite{Zehavi2017}, \textsc{Upper Edge Domination}~\cite{MonnotEtAl2021}, \textsc{Maximum Minimal Feedback Vertex Set}~\cite{DubloisEtAl2022FVS}, \textsc{Maximum Minimal Separator}~\cite{HanakaBZO19}, and \textsc{Largest Bond}~\cite{DuarteEtAl2021}.
Examples of min-max problems include \textsc{Minimum Maximal Matching}~\cite{YannakakisGavril1980} and \textsc{Minimum Maximal Independent Set}~\cite{Haviland91} (or, equivalently, \textsc{Independent Domination}~\cite{GoddardH13}).

A general problem in this family is \textsc{Maximum Minimal Hitting Set}: given a hypergraph $\mathcal H=(U,\mathcal F)$ and a positive integer $k$, determine whether $\mathcal H$ has an inclusion-wise minimal hitting set of size at least $k$~\cite{BazganEtAl2018,AraujoEtAl2023}. 
Parameterized by $k$, the problem is \textsf{W[1]}-hard and in~\textsf{XP}~\cite{BazganEtAl2018,BlaesiusEtAl2022,AraujoEtAl2023}; it is \textsf{FPT} parameterized by $k$ plus the maximum hyperedge size~\cite{AraujoEtAl2023}. 
For the clique hypergraph $\mathcal H_G=(V(G),\mcK(G))$, minimal hitting sets are exactly minimal clique transversals, so \UCT{} is \textsc{Maximum Minimal Hitting Set} restricted to clique hypergraphs implicitly represented by graphs. 
This representation is essential since $\mcK(G)$ may be exponential.

For \UCT{}, Boros et al.~\cite{BorosEtAl2025} gave an $n^{\mathcal{O}(k)}$-time algorithm and asked whether the problem is \textsf{W[1]}-hard parameterized by $k$. 
We answer this question affirmatively by a parameterized reduction from \MCC{}.
This result parallels analogous results for \textsc{Upper Domination} and \textsc{Upper Total Domination}, see~\cite{BazganEtAl2018}.

\Needspace{6\baselineskip}

\section{W[1]-Hardness}

For a positive integer $k$, we denote by $[k]$ the set $\{i\in \mathbb{Z}\colon 1\le i\le k\}$.
For a clique transversal $S$ and a vertex $s\in S$, a maximal clique $K\in\mcK(G)$ with $K\cap S=\{s\}$ is called a \emph{private clique} of $s$ with respect to $S$.

\begin{lemma}[folklore; see, e.g.,~\cite{BlaesiusEtAl2022} for a more general statement]\label{lem:private}
A clique transversal $S$ is minimal if and only if every vertex $s\in S$ has a private clique with respect to $S$.
\end{lemma}

We now show the main theorem in this paper.

\begin{theorem}\label{thm:main}
\UCT{} is \textsf{W[1]}-hard when parameterized by $k$.
\end{theorem}

We give a reduction from \MCC{}.
The input to this problem is a graph $G$, whose vertex set is partitioned into color classes, $V(G)=V_1\uplus\cdots\uplus V_r$, and the task is to determine whether we can form an $r$-clique by choosing one vertex from each color class.
This problem is \textsf{W[1]}-complete when parameterized by the number $r$ of colors~\cite{CyganEtAl2015}.
Without loss of generality, we assume that $r\ge 2$ and that each $V_i$ is an independent set.
We also assume that every $u\in V_i$ is adjacent to at least one vertex $v\in V_j$ for each $j\ne i$ (otherwise, $u$ cannot belong to any multicolored $r$-clique and can be deleted).

\paragraph{Construction.}
We construct an instance $H = (V(H), E(H))$ of \UCT{}.
We first define the vertex set of $H$, as follows.
We set $V(H)\coloneqq C\cup Q\cup A \cup X$, where
\[
\begin{aligned}
C_i&\coloneqq \{c_i^t\colon t\in[2r^2-r+1]\} &\textrm{for all }i\in [r],& \quad &C\coloneqq \bigcup_i C_i\,,\\
Q_i&\coloneqq \{q_{i,v}\colon v\in V_i\}&\textrm{for all }i\in [r],&\quad &Q\coloneqq \bigcup_iQ_i\,,\\
A_i&\coloneqq \{a_{ij}\colon j\in [r], j\neq i\}&\textrm{for all }i\in [r],&\quad &A\coloneqq \bigcup_iA_i\,,\\
X&\coloneqq \{x_e\colon e\in E(G)\}\,.&&&
\end{aligned}
\]
{As we will see later}, $C_i$ is an independent set introduced for color $i$, vertex $q_{i,v}$ represents a choice of an original vertex $v\in V_i$, and vertex $x_e$ represents an original edge.
The vertex $a_{ij}$ is an anchor for checking a pair of colors, and $a_{ij}$ and $a_{ji}$ are distinct vertices.
Furthermore, for $i,j\in [r], i\neq j$, let $X_{ij}=X_{ji}\coloneqq \{x_e\colon e\in E(V_i,V_j)\}$ be the set of edge vertices representing the original edges between colors $i$ and $j$, and define
\[
\begin{aligned}
X_i&\coloneqq \bigcup_{j\ne i}X_{ij} &&\textrm{for all }i\in [r],\\
X_i(v)&\coloneqq \{x_e\in X_i\colon v\in e\}&&\textrm{for all $i\in [r]$ and $v\in V_i$}\,.
\end{aligned}
\]
The set $X_i(v)$ consists of the vertices representing edges incident to $v\in V_i$ in the original graph,
and we have $X_i(v)\subseteq X_i=\bigcup_{v\in V_i}X_i(v)$.
For distinct vertices $u,v$ of the same color, $X_i(u)\cap X_i(v)=\emptyset$,
and for distinct colors $i,j$, $X_i\cap X_j=X_{ij}$ holds.
By assumption, each $X_i(v)$ and each $X_{ij}$ is nonempty.

The edges are defined by the following three rules.
{See Figures~\ref{fig:direct-local} and~\ref{fig:direct-pair} for a visualization of this.}
\begin{enumerate}[label=(E\arabic*)]
\item For each $i\in[r],t\in[2r^2-r+1]$,
make $D_i^{t}\coloneqq \{c_i^t\}\cup A_i\cup X_i$ a clique.
\item For each $v\in V_i$, make $q_{i,v}$ adjacent to exactly the vertices in
$C_i\cup A_i\cup(X_i\setminus X_i(v))$.
\item For each $i<j$, add the edge $a_{ij}a_{ji}$.
\end{enumerate}

This completes the construction of $H$.
{Clearly, this construction can be done in polynomial time.}
By definition, $C$ and $Q$ are independent sets.
Moreover, for each $i\in[r],v\in V_i,t\in[2r^2-r+1]$, define
\[W_{i,v}^t\coloneqq \{q_{i,v},c_i^t\}\cup A_i\cup(X_i\setminus X_i(v)),\]
and for each $i,j\in[r]$ with $i\ne j$, define
\[Y_{ij}\coloneqq \{a_{ij},a_{ji}\}\cup X_{ij}\,.\]
Then $W_{i,v}^t$ and $Y_{ij}$ are cliques.
Figure~\ref{fig:direct-local} illustrates the basic structure for color $i\in [r]$.
Figure~\ref{fig:direct-pair} illustrates the overlap between two colors.

\begin{figure}[!tbp]
\centering
\begin{tikzpicture}[x=0.98cm,y=0.82cm]
  \fill[UCTblue!4,rounded corners=2mm] (-5.12,-1.02) rectangle (6.50,1.02);
  \fill[UCTamber!6,rounded corners=2mm]
    (-5.12,-2.85)--(-3.25,-2.85)--(-3.25,-0.80)--(6.30,-0.80)
    --(6.30,0.80)--(-3.45,0.80)--(-3.45,-1.85)--(-5.12,-1.85)--cycle;
  \fill[UCTteal!5,rounded corners=2mm] (-7.02,-0.46) rectangle (3.15,0.46);
  \draw[UCTgray,dashed,rounded corners=2mm] (-6.85,-1.6) rectangle (-5.55,1.6);
  \draw[UCTgray,dashed,rounded corners=2mm] (-5.0,-2.78) rectangle (-3.7,0.6);
  \node[uctlabel,text=UCTgray] at (-6.5,-2.25) {$Q_i$\\independent set};
  \node[uctlabel,text=UCTgray] at (-4.35,-3.55) {$C_i$\\independent set};
  \node[uctlabel] at (-6.2,1.13) {$\vdots$};
  \node[uctlabel] at (-6.2,-1.13) {$\vdots$};
  \node[uctchoice] at (-6.2,0) {$q_{i,v}$};
  \node[uctcenter] at (-4.35,0) {$c_i^1$};
  \node[uctlabel] at (-4.35,-1.38) {$\vdots$};
  \node[uctcenter,draw=UCTamber,fill=UCTamber!9,text=UCTamber]
    at (-4.35,-2.35) {$c_i^t$};
  \node[uctbox,draw=UCTgray,fill=white,minimum width=1.55cm,minimum height=0.64cm]
    at (-2.35,0) {$A_i$};
  \draw[UCTpurple,rounded corners=1mm] (-1.15,-0.62) rectangle (6.15,0.62);
  \node[uctbox,draw=UCTpurple,fill=white,minimum width=3.85cm,minimum height=0.62cm]
    at (0.94,0) {$X_i\setminus X_i(v)$};
  \node[uctbox,draw=UCTpurple,fill=white,minimum width=2.30cm,minimum height=0.62cm]
    at (4.72,0) {$X_i(v)$};
  \node[uctlabel,text=UCTpurple] (XiLabel) at (0.7,-1.55) {$X_i$};
  \draw[UCTblue,line width=1.05pt,rounded corners=2mm]
    (-5.12,-1.02) rectangle (6.50,1.02);
  \draw[UCTamber,line width=1.05pt,rounded corners=2mm]
    (-5.12,-2.85)--(-3.25,-2.85)--(-3.25,-0.80)--(6.30,-0.80)
    --(6.30,0.80)--(-3.45,0.80)--(-3.45,-1.85)--(-5.12,-1.85)--cycle;
  \draw[UCTteal,line width=1.05pt,rounded corners=2mm]
    (-7.02,-0.46) rectangle (3.15,0.46);
  \node[uctlabel,text=UCTblue] (DiOneLabel) at (4.25,1.95) {$D_i^1=\{c_i^1\}\cup A_i\cup X_i$};
  \node[uctlabel,text=UCTteal] (WiOneLabel) at (-5.3,1.97) {$W_{i,v}^1=\{q_{i,v},c_i^1\}\cup A_i\cup\bigl(X_i\setminus X_i(v)\bigr)$};
  \draw[uctarrow,draw=UCTblue] (DiOneLabel.south)--(4.25,1.02);
  \draw[uctarrow,draw=UCTteal] (WiOneLabel.south)--(-5.3,0.46);
  \node[uctlabel,text=UCTamber] at (-.3,-2.30) {$D_i^t=\{c_i^t\}\cup A_i\cup X_i$};
  \draw[uctarrow,draw=UCTamber] (-2.4,-2.30)--(-3.25,-2.30);
  \draw[uctarrow,draw=UCTpurple] (XiLabel.north)--(0.7,-0.62);
  \draw[uctarrow,draw=UCTpurple] (4.72,-1.35)--(4.72,-0.4);
  \node[uctlabel,text=UCTpurple] at (4.55,-1.75) {nonadjacent to $q_{i,v}$};
\end{tikzpicture}
\caption{The local structure for color $i\in [r]$.
The blue, orange, and green outlines represent $D_i^1$, $D_i^t$, and $W_{i,v}^1$, respectively.
The purple outline represents $X_i$.
The sets $Q_i,C_i$ enclosed by dashed outlines are independent sets.}
\label{fig:direct-local}
\end{figure}
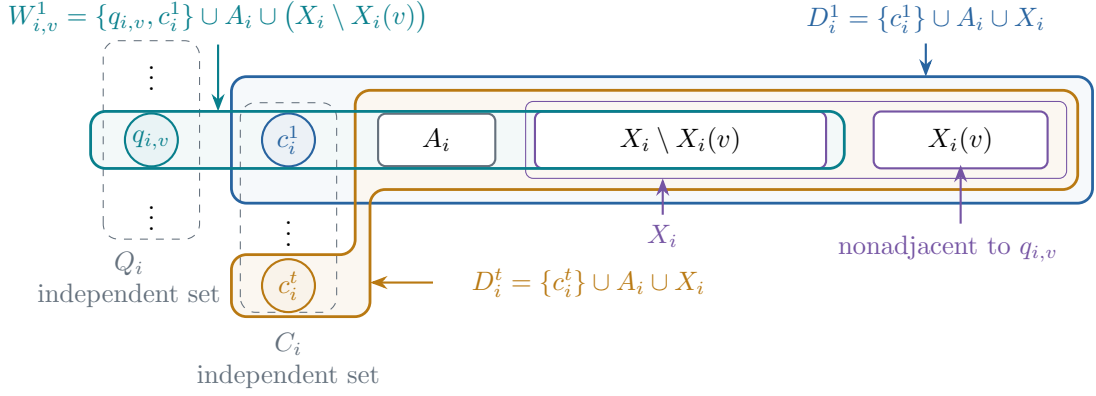

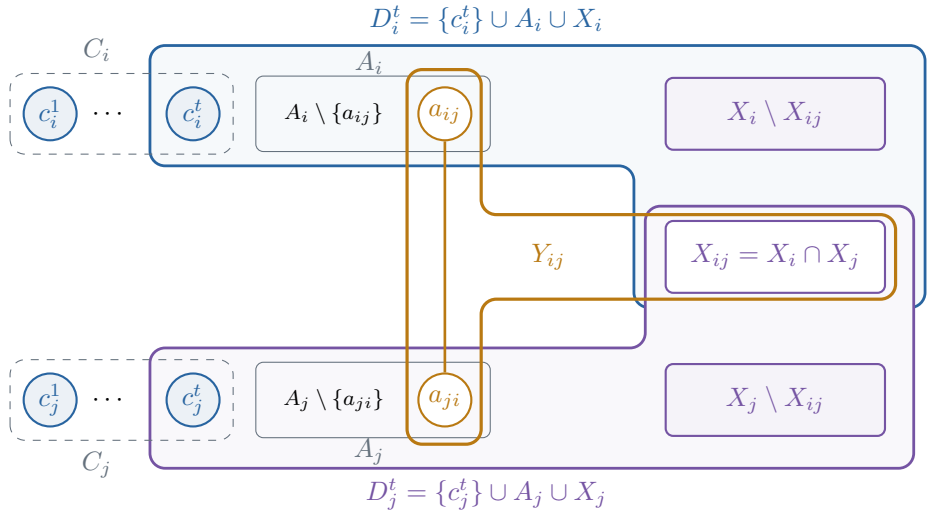
\begin{figure}[htbp]
\centering
\begin{tikzpicture}[x=1cm,y=0.86cm]
  \path[fill=UCTblue!4,draw=UCTblue,line width=1.05pt,rounded corners=2mm]
    (-4.55,1.4)--(1.85,1.4)--(1.85,-0.78)--(5.7,-0.78)
    --(5.7,3.25)--(-4.55,3.25)--cycle;
  \path[fill=UCTpurple!4,draw=UCTpurple,line width=1.05pt,rounded corners=2mm]
    (-4.55,-1.4)--(2.0,-1.4)--(2.0,0.78)--(5.55,0.78)
    --(5.55,-3.25)--(-4.55,-3.25)--cycle;
  \foreach \yy in {2.2,-2.2} {
    \draw[UCTgray,dashed,rounded corners=2mm] (-6.4,\yy-0.62) rectangle (-3.45,\yy+0.62);
    \draw[UCTgray,rounded corners=1mm] (-3.15,\yy-0.58) rectangle (-0.05,\yy+0.58);
  }
  \node[uctcenter] at (-5.87,2.2) {$c_i^1$};
  \node[uctlabel] at (-5.08,2.2) {$\cdots$};
  \node[uctcenter] at (-3.98,2.2) {$c_i^t$};
  \node[uctcenter] at (-5.87,-2.2) {$c_j^1$};
  \node[uctlabel] at (-5.08,-2.2) {$\cdots$};
  \node[uctcenter] at (-3.98,-2.2) {$c_j^t$};
  \node[uctlabel,text=UCTgray] at (-5.25,3.19) {$C_i$};
  \node[uctlabel,text=UCTgray] at (-5.25,-3.19) {$C_j$};
  \node[uctlabel,font=\scriptsize] at (-2.12,2.2) {$A_i\setminus\{a_{ij}\}$};
  \node[uctlabel,font=\scriptsize] at (-2.12,-2.2) {$A_j\setminus\{a_{ji}\}$};
  \node[uctlabel,text=UCTgray] at (-1.65,2.99) {$A_i$};
  \node[uctlabel,text=UCTgray] at (-1.65,-2.99) {$A_j$};
  \node[uctvertex,draw=UCTamber,fill=white,text=UCTamber] (anchorij) at (-0.65,2.2) {$a_{ij}$};
  \node[uctvertex,draw=UCTamber,fill=white,text=UCTamber] (anchorji) at (-0.65,-2.2) {$a_{ji}$};
  \draw[UCTamber,line width=1pt] (anchorij)--(anchorji);
  \node[uctcode,minimum width=2.90cm,minimum height=0.95cm] at (3.72,2.2) {$X_i\setminus X_{ij}$};
  \node[uctcode,minimum width=2.90cm,minimum height=0.95cm] at (3.72,-2.2) {$X_j\setminus X_{ij}$};
  \node[uctcode,fill=white,minimum width=2.90cm,minimum height=0.95cm] at (3.72,0) {$X_{ij}=X_i\cap X_j$};
  \path[draw=UCTamber,line width=1.2pt,rounded corners=2mm]
    (-1.15,2.88)--(-0.17,2.88)--(-0.17,0.65)--(5.31,0.65)
    --(5.31,-0.65)--(-0.17,-0.65)--(-0.17,-2.88)--(-1.15,-2.88)--cycle;
  \node[uctlabel,text=UCTamber] at (0.72,0) {$Y_{ij}$};
  \node[uctlabel,text=UCTblue] at (-0.1,3.65) {$D_i^t=\{c_i^t\}\cup A_i\cup X_i$};
  \node[uctlabel,text=UCTpurple] at (-0.1,-3.65) {$D_j^t=\{c_j^t\}\cup A_j\cup X_j$};
\end{tikzpicture}
\caption{The overlap of cliques between two colors.
The blue outline $D_i^t$ and the purple outline $D_j^t$ share $X_{ij}$.
The orange outline represents $Y_{ij}=\{a_{ij},a_{ji}\}\cup X_{ij}$.}
\label{fig:direct-pair}
\end{figure}

\begin{lemma}\label{lem:structure}
The following statements hold for the graph $H$.
\begin{enumerate}[label=(\alph*)]
\item $D_i^{t},W_{i,v}^t,Y_{ij}$ are maximal cliques.
Every maximal clique containing $q_{i,v}$ is $W_{i,v}^t$ for some $t\in[2r^2-r+1]$,
and every maximal clique containing $c_i^t$ is either $D_i^{t}$ or $W_{i,v}^t$ for some $v\in V_i$.
\item Every maximal clique that contains no vertex of $C\cup Q$ and contains the anchor $a_{ij}$ is $Y_{ij}$.
\item If a maximal clique $K$ satisfies $K\cap Q = \emptyset$ and $K\cap X_{ij} \neq \emptyset$, then $X_{ij}\subseteq K$.
\end{enumerate}
\end{lemma}

\begin{proof}
(a) By the construction of $H$, the neighborhood of $q_{i,v}$ is $A_i\cup(X_i\setminus X_i(v))\cup C_i$.
For each $t\in[2r^2-r+1]$, the set $A_i\cup(X_i\setminus X_i(v))\cup \{c^t_i\}$ forms a clique.
On the other hand, since $C_i$ is an independent set, every maximal clique containing $q_{i,v}$ must be
$W_{i,v}^t=\{q_{i,v},c_i^t\}\cup A_i\cup(X_i\setminus X_i(v))$.
Furthermore, since $N_H(q_{i,v})\cap N_H(c_i^t)=A_i\cup(X_i\setminus X_i(v))\subseteq W_{i,v}^t$,
no vertex can be added to $W_{i,v}^t$ to form a larger clique, and hence $W_{i,v}^t$ is maximal.

We also have $N_H(c_i^t)=A_i\cup X_i\cup Q_i$.
Any vertex that extends $D_i^t$ must be adjacent to $c_i^t$ and lie outside $D_i^t$,
and hence must belong to $N_H(c_i^t)\setminus D_i^t=Q_i$.
However, each $q_{i,v}\in Q_i$ is nonadjacent to every vertex of the nonempty set $X_i(v)\subseteq D_i^t$,
and thus cannot extend $D_i^t$.
Therefore, $D_i^t$ is maximal.

Next, consider an arbitrary maximal clique $K$ containing $c_i^t$.
Since every other vertex of $K$ is adjacent to $c_i^t$,
we have $K\subseteq\{c_i^t\}\cup A_i\cup X_i\cup Q_i$.
If $K\cap Q_i=\emptyset$, then $K\subseteq D_i^t$.
Since $D_i^t$ itself is a clique, if $K\subset D_i^t$, then a vertex of $D_i^t\setminus K$ can be added to $K$,
contradicting the maximality of $K$.
Thus, $K=D_i^t$.
On the other hand, if $K\cap Q_i\ne\emptyset$, then, as $Q_i$ is an independent set,
there is a unique $v\in V_i$ such that $K\cap Q_i=\{q_{i,v}\}$.
Since $q_{i,v}$ is nonadjacent to every vertex of $X_i(v)$, the clique $K$ cannot contain any vertex of $X_i(v)$.
Therefore,
\[
K\subseteq\{c_i^t,q_{i,v}\}\cup A_i\cup(X_i\setminus X_i(v))=W_{i,v}^t.
\]
Since $W_{i,v}^t$ itself is a clique, the maximality of $K$ implies $K=W_{i,v}^t$, as before.
Thus, every maximal clique containing $c_i^t$ is either $D_i^t$ or $W_{i,v}^t$.

By construction, $N_H(a_{ij})\cap N_H(a_{ji})=X_{ij}$.
Any vertex that extends $Y_{ij}$ must be adjacent to both anchors,
but their common neighborhood $X_{ij}$ is already contained in $Y_{ij}$.
Thus, no vertex outside $Y_{ij}$ can extend this clique, and hence it is maximal.

\par\medskip\noindent(b) Let $K$ be a maximal clique in $H$ that contains no vertex of $C\cup Q$ and contains the anchor $a_{ij}$.
Since $a_{ij}\in K$, every other vertex of $K$ is adjacent to $a_{ij}$.
The vertices outside $C\cup Q$ that are adjacent to $a_{ij}$ can only lie in $(A_i\setminus\{a_{ij}\})\cup X_i\cup\{a_{ji}\}$, and, hence, $K\subseteq A_i\cup X_i\cup\{a_{ji}\}$.
Thus, if $a_{ji}\notin K$, then $K\subseteq A_i\cup X_i$, and we can add $c_i^1\notin K$, contradicting maximality.
Therefore, $a_{ji}\in K$.
Since every other vertex of $K$ is adjacent to the two vertices $a_{ij},a_{ji}$, we have
\[K\setminus\{a_{ij},a_{ji}\}\subseteq N_H(a_{ij})\cap N_H(a_{ji})=X_{ij}\,.\]
Thus, $K\subseteq\{a_{ij},a_{ji}\}\cup X_{ij}=Y_{ij}$.
Since $Y_{ij}$ is a maximal clique, the maximality of $K$ implies $K=Y_{ij}$.

\par\medskip\noindent(c) Let $K$ be a maximal clique in $H$ satisfying $K\cap Q = \emptyset$ and $K\cap X_{ij} \neq \emptyset$.
The closed neighborhood of any $x\in X_{ij}$ in $H-Q$ is
\[N_{H-Q}[x]=C_i\cup C_j\cup A_i\cup A_j\cup X_i\cup X_j\,,\]
and, hence, all vertices of $X_{ij}$ have the same closed neighborhood in $H-Q$.
By the assumption that $K\cap X_{ij}\ne\emptyset$, we can choose a vertex $x\in K\cap X_{ij}$.
Since $K$ is a clique containing no vertex of~$Q$, we have $K\subseteq N_{H-Q}[x]$.
If there is a vertex $y\in X_{ij}\setminus K$, then the equality of the closed neighborhoods gives
$K\subseteq N_{H-Q}[x]=N_{H-Q}[y]$.
Since $y\notin K$, this means that $y$ is adjacent to every vertex of $K$,
and hence $K\cup\{y\}$ is a clique, contradicting the maximality of $K$.
Therefore, $X_{ij}\setminus K=\emptyset$, that is, $X_{ij}\subseteq K$.
\end{proof}

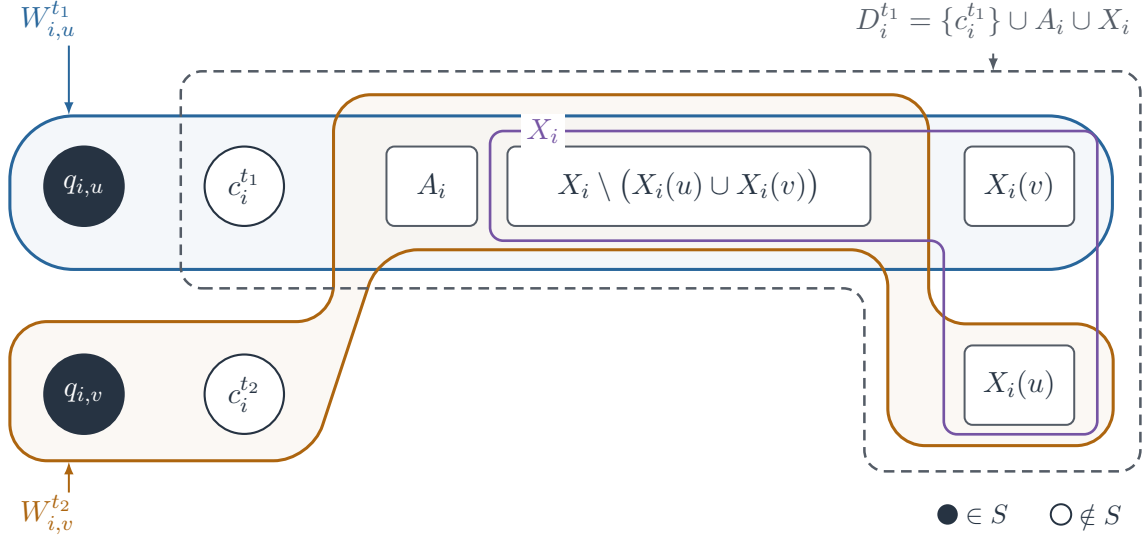
\begin{figure}[!tbp]
\centering  
\begin{tikzpicture}[
  x=1cm, y=1cm, font=\normalsize,
  selected/.style={circle, draw=LFourInk, fill=LFourInk,
    text=white, line width=.8pt, minimum size=10.5mm, inner sep=1pt},
  unselected/.style={circle, draw=LFourInk, fill=white,
    line width=.8pt, minimum size=10.5mm, inner sep=1pt},
  block/.style={draw=LFourGray, fill=white, rounded corners=3pt,
    line width=.75pt, minimum height=10.5mm, inner sep=5pt},
  leader/.style={-{Latex[length=1.7mm]}, line width=.75pt},
  every node/.append style={text=LFourInk}
]
\def\WuPath{(0,1.65) rectangle (14.58,3.68)}
\def\WvPath{
  (0,-.88) -- (3.98,-.88) -- (4.90,1.91)
  -- (11.61,1.91) -- (11.61,-.68) -- (14.59,-.68)
  -- (14.59,.93) -- (12.15,.93) -- (12.15,3.96)
  -- (4.28,3.96) -- (4.28,.96) -- (0,.96) -- cycle
}
\def\DPath{
  (2.26,1.40) -- (11.30,1.40) -- (11.30,-1.02)
  -- (14.95,-1.02) -- (14.95,4.27) -- (2.26,4.27) -- cycle
}
\path[fill=LFourBlue!5, rounded corners=24pt] \WuPath;
\path[fill=LFourOrange!7, fill opacity=.60, rounded corners=14pt] \WvPath;
\draw[draw=LFourBlue, line width=1.15pt, rounded corners=24pt] \WuPath;
\draw[draw=LFourOrange, line width=1.15pt, rounded corners=14pt] \WvPath;
\draw[draw=LFourGray, line width=.95pt,
      dash pattern=on 4pt off 2.5pt, rounded corners=8pt] \DPath;
\node[selected, text=white] (qu) at (.98,2.75) {$q_{i,u}$};
\node[unselected] (cu) at (3.10,2.75) {$c_i^{t_1}$};
\node[selected, text=white] (qv) at (.98,0) {$q_{i,v}$};
\node[unselected] (cv) at (3.10,0) {$c_i^{t_2}$};
\node[block, minimum width=12mm] (A) at (5.58,2.75) {$A_i$};
\node[block, minimum width=48mm] (R) at (8.98,2.75)
  {$X_i\setminus\bigl(X_i(u)\cup X_i(v)\bigr)$};
\node[block, minimum width=14.5mm] (Xv) at (13.35,2.75) {$X_i(v)$};
\node[block, minimum width=14.5mm] (Xu) at (13.35,.12) {$X_i(u)$};

\draw[draw=UCTpurple, line width=1pt, rounded corners=5pt]
  (6.35,3.48) -- (14.38,3.48) -- (14.38,-.53)
  -- (12.35,-.53) -- (12.35,2.03) -- (6.35,2.03) -- cycle;
\node[text=UCTpurple, fill=white, inner sep=2pt] at (7.05,3.48) {$X_i$};
\node[anchor=west, text=LFourBlue] at (0,4.94)
  {$W_{i,u}^{t_1}$};
\draw[leader, LFourBlue] (.78,4.61) -- (.78,3.69);
\node[anchor=east, inner xsep=0pt, text=LFourGray]
  (Ddefinition) at (14.83,4.94)
  {$D_i^{t_1}=\{c_i^{t_1}\}\cup A_i\cup X_i$};
\draw[leader, LFourGray, shorten <=2pt]
  (Ddefinition.south) -- (Ddefinition.south |- 0,4.27);

\node[anchor=west, text=LFourOrange] at (0,-1.59)
  {$W_{i,v}^{t_2}$};
\draw[leader, LFourOrange] (.78,-1.30) -- (.78,-.89);

\fill[LFourInk] (12.40,-1.59) circle (1.35mm);
\node[anchor=west, font=\small] at (12.5,-1.59) {$\in S$};
\draw[LFourInk, fill=white, line width=.75pt] (13.90,-1.59) circle (1.35mm);
\node[anchor=west, font=\small] at (14.0,-1.59) {$\notin S$};

\end{tikzpicture}
\caption{An illustration of the proof by contradiction of $|S\cap Q_i|\le1$ in \Cref{lem:bound}.
Since the blue and orange cliques are private, the clique $D_i^{t_1}$ enclosed by the dashed outline
does not intersect $S$, contradicting the clique transversal property.}
\label{lem:bound-fig}
\end{figure}

\begin{lemma}\label{lem:bound}
Let $S$ be a minimal clique transversal in $H$.
Then
\begin{enumerate}[label=(\alph*)]
\item $|S\cap Q_i|\le 1$ for all $i\in [r]$,
    \item $|S\cap X_{ij}|\le 2$ for all $i,j\in [r]$, $i\neq j$,
    \item $|S\setminus C|\le 2r^2-r$.
\end{enumerate}
\end{lemma}
\begin{proof}
(a) \Cref{lem:bound-fig} illustrates the proof of $|S\cap Q_i|\le 1$.
Suppose that $q_{i,u},q_{i,v}\in S$ for distinct $u,v\in V_i$.
By Lemma~\ref{lem:structure}(a), their private cliques are $W_{i,u}^{t_1},W_{i,v}^{t_2}$ for some $t_1,t_2\in [2r^2-r+1]$.
Since $X_i(u)\cap X_i(v)=\emptyset$, we have
\[A_i\cup X_i\subseteq \bigl(W_{i,u}^{t_1}\setminus\{q_{i,u}\}\bigr)\cup \bigl(W_{i,v}^{t_2}\setminus\{q_{i,v}\}\bigr)\,.\]
By the definition of a private clique,
\[W_{i,u}^{t_1}\cap S=\{q_{i,u}\},\qquad
W_{i,v}^{t_2}\cap S=\{q_{i,v}\}\,.\]
Thus, the set obtained from each clique by removing the corresponding $q$ vertex is disjoint from $S$.
By the above inclusion, $(A_i\cup X_i)\cap S=\emptyset$.
Moreover, since $c_i^{t_1}\in W_{i,u}^{t_1}\setminus\{q_{i,u}\}$, the same private-clique property implies $c_i^{t_1}\notin S$.
Together with $D_i^{t_1}=\{c_i^{t_1}\}\cup A_i\cup X_i$, this gives
\[D_i^{t_1}\cap S
=\bigl(\{c_i^{t_1}\}\cap S\bigr)\cup\bigl((A_i\cup X_i)\cap S\bigr)
=\emptyset\,.\]
However, by Lemma~\ref{lem:structure}(a), $D_i^{t_1}$ is a maximal clique,
and the clique transversal $S$ must intersect it.
This is a contradiction. 

\par\medskip\noindent(b)
Suppose that there are distinct vertices $x,y,z\in S\cap X_{ij}$.
Their private cliques intersect $Q$.
Indeed, suppose that a maximal clique $K$ containing $x$ does not intersect $Q$.
By Lemma~\ref{lem:structure}(c), $K$ contains all vertices of $X_{ij}$, and hence also contains $y,z$.
Since $y,z\in S$, we cannot have $K\cap S=\{x\}$,
and hence $K$ is not a private clique of $x$.
Thus, every private clique of $x$ must intersect $Q$.
The same holds for $y,z$.
For each $p\in\{x,y,z\}$, denote its private clique by $P_p$.
Since $P_p\cap Q\ne\emptyset$, we can choose a vertex $q_{k,u}\in P_p\cap Q$ (with $k\in [r]$ and $u\in V_k$).
The vertices $p$ and $q_{k,u}$ are adjacent as they belong to the same clique.
By construction, no vertex of $X_{ij}$ is adjacent to a vertex of $Q$ outside $Q_i\cup Q_j$,
and hence $k\in\{i,j\}$.
Furthermore, by Lemma~\ref{lem:structure}(a), every maximal clique containing $q_{k,u}$ is of the form $W_{k,u}^t$ for some $t\in [2r^2-r+1]$.
Therefore, each $P_p$ is a $W$ clique on either the $i$ side or the $j$ side.
By the pigeonhole principle and the symmetry between $i,j$, after relabelling $x,y,z$ and interchanging $i,j$ if necessary,
we may write the private cliques of $x,y$ as
\[P_x=W_{i,u}^{t_1},\qquad P_y=W_{i,v}^{t_2}\,.\]
We distinguish two cases, $u\ne v$ and $u=v$ (see Figure~\ref{fig:bound-x-private}).

\smallskip
Assume first that $u\ne v$.
Since an original edge has only one endpoint of color $i$, we have $X_i(u)\cap X_i(v)=\emptyset$.
Since $P_x\cap X_i=X_i\setminus X_i(u)$ and $P_y\cap X_i=X_i\setminus X_i(v)$, we have
\[(P_x\cup P_y)\cap X_i
=(X_i\setminus X_i(u))\cup(X_i\setminus X_i(v))=X_i\,.\]
Thus, $z\in S\cap X_{ij}\subseteq S\cap X_i$ belongs to at least one of $P_x,P_y$.
This contradicts either $P_x\cap S=\{x\}$ or $P_y\cap S=\{y\}$.

\smallskip
Assume now that $u=v$.
By the definition of $W$ cliques and the equality $u=v$, we have
\[P_x\cap X_i=P_y\cap X_i=X_i\setminus X_i(u)\,.\]
Since $P_x$ is a private clique of $x$, we have $x\in P_x$,
and we also have $x\in X_{ij}\subseteq X_i$.
Thus, the above equality gives $x\in P_y\cap X_i$.
Moreover, since $x\in S$, we have $x\in P_y\cap S$.
However, $P_y$ is a private clique of $y$, and hence $P_y\cap S=\{y\}$,
contradicting $x\ne y$.

\smallskip
Both cases lead to a contradiction, and hence $|S\cap X_{ij}|\le2$.

\par\medskip\noindent(c)
It remains to show that $|S\setminus C|\le 2r^2-r$.
By construction, $V(H)\setminus C$ is partitioned into $Q,A,X$.
Moreover, $Q$ is partitioned into the pairwise disjoint sets $Q_i$ $(i\in[r])$,
and $X$ is partitioned into the pairwise disjoint sets $X_{ij}$ $(1\le i<j\le r)$.
Therefore, by the bounds $|S\cap Q_i|\le1$ and $|S\cap X_{ij}|\le2$ proved above,
together with $|S\cap A|\le|A|=r(r-1)$, we obtain
\[
\begin{aligned}
|S\setminus C|
&=\sum_{i=1}^r|S\cap Q_i|+\sum_{1\le i<j\le r}|S\cap X_{ij}|+|S\cap A|\\
&\le r+2\binom r2+r(r-1)=2r^2-r.
\end{aligned}
\]
This completes the proof of the lemma.
\end{proof}

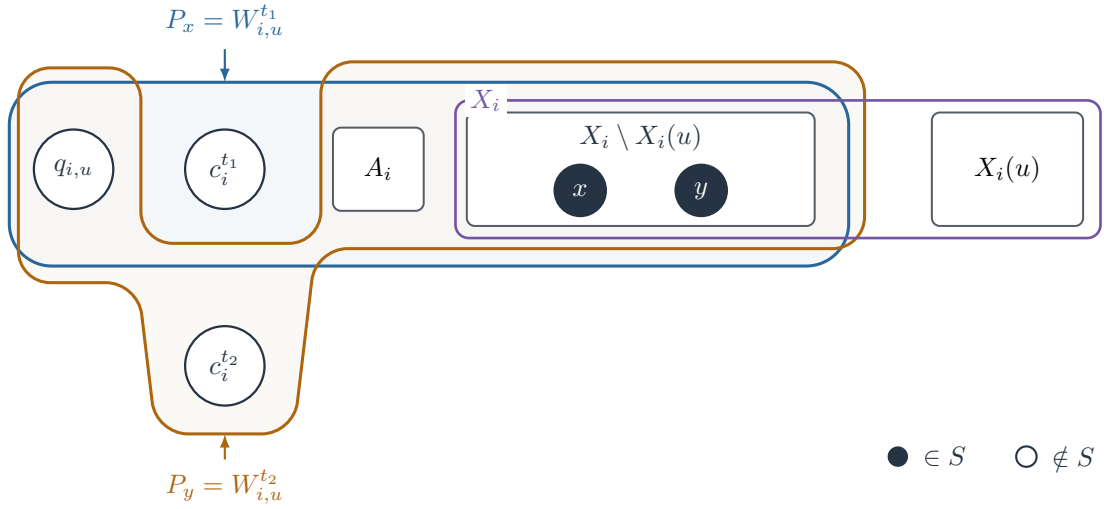
\begin{figure}[!tbp]
\centering
\begin{tikzpicture}[
  x=1cm,y=1cm,font=\small,
  unselected/.style={circle,draw=LFourInk,fill=white,text=LFourInk,
    line width=.85pt,minimum size=10.5mm,inner sep=1pt},
  selected/.style={circle,draw=LFourInk,fill=LFourInk,text=white,
    line width=.8pt,minimum size=6.8mm,inner sep=1pt},
  block/.style={draw=LFourGray,fill=white,rounded corners=3pt,
    line width=.75pt,inner sep=4pt},
  region/.style={line width=1.15pt},
  leader/.style={-{Latex[length=1.8mm]},line width=.8pt},
  note/.style={align=center,text=LFourInk}
]
\begin{scope}
  \node[anchor=west,font=\normalsize,inner xsep=0pt] at (0,5.2)
    {(a) Case $u\ne v$: the two private cliques cover all of $X_i$};

  \def\PxPath{(0,1.22) rectangle (14.75,3.65)}
  \def\PyPath{
    (0,-1.22) -- (3.65,-1.22) -- (4.12,1.43)
    -- (11.20,1.43) -- (11.20,-1.22) -- (14.72,-1.22)
    -- (14.72,.99) -- (11.67,.99) -- (11.67,3.92)
    -- (3.98,3.92) -- (3.98,.92) -- (0,.92) -- cycle
  }
  \path[fill=LFourBlue!5,rounded corners=16pt] \PxPath;
  \path[fill=LFourOrange!6,fill opacity=.65,rounded corners=12pt] \PyPath;
  \draw[LFourBlue,region,rounded corners=16pt] \PxPath;
  \draw[LFourOrange,region,rounded corners=12pt] \PyPath;
  \node[unselected] at (.85,2.50) {$q_{i,u}$};
  \node[unselected] at (2.85,2.50) {$c_i^{t_1}$};
  \node[unselected] at (.85,-.10) {$q_{i,v}$};
  \node[unselected] at (2.85,-.10) {$c_i^{t_2}$};

  \node[block,minimum width=12mm,minimum height=11mm]
    at (4.88,2.50) {$A_i$};
  \node[block,minimum width=46mm,minimum height=15mm]
    at (8.35,2.50) {};
  \node[note] at (8.35,2.92) {$X_i\setminus\bigl(X_i(u)\cup X_i(v)\bigr)$};
  \draw[UCTteal,fill=UCTteal!8,dashed,line width=1pt,rounded corners=3pt]
    (6.35,1.84) rectangle (10.35,2.60);
  \node[text=UCTteal,font=\scriptsize] at (8.35,2.22) {$X_{ij}\setminus\bigl(X_i(u)\cup X_i(v)\bigr)$};

  \node[block,minimum width=20mm,minimum height=15mm]
    at (13.20,2.50) {};
  \node[note] at (13.20,2.92) {$X_i(v)$};
  \draw[UCTteal,fill=UCTteal!8,dashed,line width=1pt,rounded corners=3pt]
    (12.35,1.84) rectangle (14.05,2.60);
  \node[selected] at (13.20,2.22) {$x$};
  \node[block,minimum width=20mm,minimum height=15mm]
    at (13.20,-.10) {};
  \node[note] at (13.20,.32) {$X_i(u)$};
  \draw[UCTteal,fill=UCTteal!8,dashed,line width=1pt,rounded corners=3pt]
    (12.35,-.76) rectangle (14.05,0);
  \node[selected] at (13.20,-.38) {$y$};

  \draw[UCTpurple,line width=1pt,rounded corners=5pt]
    (5.90,3.41) -- (14.43,3.41) -- (14.43,-.96)
    -- (12.00,-.96) -- (12.00,1.59) -- (5.90,1.59) -- cycle;
  \node[text=UCTpurple,fill=white,inner sep=2pt]
    at (6.30,3.41) {$X_i$};
 \node[align=center,text=UCTteal,fill=white,inner sep=4pt] (aXij) at (8.35,.10)
    {{\Large $X_{ij}$}};
  \draw[leader,UCTteal] (aXij.north) -- (8.35,1.84);
  \draw[leader,UCTteal] (aXij.north east) -- (12.35,1.95);
  \draw[leader,UCTteal] (aXij.east) -- (12.35,-.38);

  \node[text=LFourBlue] (aPx) at (1.60,4.45)
    {$P_x=W_{i,u}^{t_1}$};
  \draw[leader,LFourBlue] (aPx.south) -- (1.60,3.65);
  \node[anchor=east,note,inner xsep=0pt] at (14.75,4.45)
    {$x,y,z\in S\cap X_{ij}$};
  \node[text=LFourOrange] (aPy) at (1.60,-1.90)
    {$P_y=W_{i,v}^{t_2}$};
  \draw[leader,LFourOrange] (aPy.north) -- (1.60,-1.22);
\end{scope}

\begin{scope}[yshift=-8.7cm]
  \node[anchor=west,font=\normalsize,inner xsep=0pt] at (0,5.2)
    {(b) Case $u=v$: the two private cliques coincide within $X_i$};

  \def\PxPath{(0,1.22) rectangle (11.10,3.65)}
  \def\PyPath{
    (.12,1.00) -- (1.70,1.00) -- (1.94,-1.00)
    -- (3.76,-1.00) -- (4.05,1.45) -- (11.31,1.45)
    -- (11.31,3.92) -- (4.12,3.92) -- (4.12,1.52)
    -- (1.74,1.52) -- (1.74,3.84) -- (.12,3.84) -- cycle
  }
  \path[fill=LFourBlue!5,rounded corners=16pt] \PxPath;
  \path[fill=LFourOrange!6,fill opacity=.65,rounded corners=12pt] \PyPath;
  \draw[LFourBlue,region,rounded corners=16pt] \PxPath;
  \draw[LFourOrange,region,rounded corners=12pt] \PyPath;

  \node[unselected] at (.85,2.50) {$q_{i,u}$};
  \node[unselected] at (2.85,2.50) {$c_i^{t_1}$};
  \node[unselected] at (2.85,-.10) {$c_i^{t_2}$};
  \node[block,minimum width=12mm,minimum height=11mm]
    at (4.88,2.50) {$A_i$};

  \node[block,minimum width=46mm,minimum height=15mm]
    at (8.35,2.50) {};
  \node[note] at (8.35,2.92) {$X_i\setminus X_i(u)$};
  \node[selected] at (7.55,2.22) {$x$};
  \node[selected] at (9.15,2.22) {$y$};
  \node[block,minimum width=20mm,minimum height=15mm]
    at (13.20,2.50) {$X_i(u)$};

  \draw[UCTpurple,line width=1pt,rounded corners=5pt]
    (5.90,1.59) rectangle (14.43,3.41);
  \node[text=UCTpurple,fill=white,inner sep=2pt]
    at (6.30,3.41) {$X_i$};

  \node[text=LFourBlue] (bPx) at (2.85,4.45)
    {$P_x=W_{i,u}^{t_1}$};
  \draw[leader,LFourBlue] (bPx.south) -- (2.85,3.65);
  \node[text=LFourOrange] (bPy) at (2.85,-1.70)
    {$P_y=W_{i,u}^{t_2}$};
  \draw[leader,LFourOrange] (bPy.north) -- (2.85,-1.00);
\end{scope}

\fill[LFourInk] (11.75,-10.0) circle (1.35mm);
\node[anchor=west,note] at (11.95,-10.0) {$\in S$};
\draw[LFourInk,fill=white,line width=.8pt]
  (13.45,-10.0) circle (1.35mm);
\node[anchor=west,note] at (13.65,-10.0) {$\notin S$};
\end{tikzpicture}
\caption{An illustration of the proof by contradiction of $|S\cap X_{ij}|\le2$.
In (a), the three parts enclosed by green dashed outlines represent the intersections of $X_{ij}$ with the respective regions;
their union is $X_{ij}$ (some parts may be empty).
We have $x,y\in X_{ij}$, and the existence of $z\in X_{ij}$ contradicts the private-clique property.
In (b), $P_x$ and $P_y$ coincide within $X_i$, so $x$ also belongs to $P_y$,
contradicting the private-clique property of $P_y$.}
\label{fig:bound-x-private}
\end{figure}

\begin{lemma}
\label{lem:correctness}
There is a multicolored clique $\{v_1,\ldots,v_r\}$ in $G$ if and only if there is a minimal clique transversal $S$ in $H$ with $|S|\ge r(2r^2-r+1)$.
\end{lemma}

\begin{proof}
\textbf{Necessity.}
Suppose that $G$ has a multicolored clique $\{v_1,\ldots,v_r\}$ with $v_i\in V_i$.
For each pair of colors $i<j$, choose the corresponding edge vertex $x_{v_iv_j}$ and define
\[S\coloneqq C\cup\{x_{v_iv_j}\colon 1\le i<j\le r\}\,.\]
We first show that this set is a clique transversal and then show that it is minimal.

First, let $K$ be a maximal clique in $H$.
If $K$ contains a vertex of $C$, then $K\cap S\ne\emptyset$ since $C\subseteq S$.
If $K$ contains a vertex $q_{i,v}$ of $Q$, then, by Lemma~\ref{lem:structure}(a),
we have $K=W_{i,v}^t$, and hence $c_i^t\in K\cap S$.
Now consider the case where $K\cap(C\cup Q)=\emptyset$.
If $K$ contains an anchor $a_{ij}$, then, by Lemma~\ref{lem:structure}(b),
we have $K=Y_{ij}$ and $x_{v_iv_j}\in K\cap S$.
In the remaining case, $K\subseteq X$.
Since $K$ is nonempty, there is a pair of colors $i<j$ such that $K\cap X_{ij}\ne\emptyset$.
By Lemma~\ref{lem:structure}(c), we have $X_{ij}\subseteq K$,
and hence $x_{v_iv_j}\in K\cap S$ in this case as well.
Thus, $S$ intersects every maximal clique.

Next, we show that every vertex of $S$ has a private clique.
Fix an arbitrary color $i$ and consider the vertex $v_i\in V_i$ of the multicolored clique.
By the definition of $S$, the only vertices in $S\cap X_i$ are
$x_{v_iv_j}$ $(j\ne i)$.
For any color $j\ne i$, one endpoint of the edge $v_iv_j$ is this fixed vertex $v_i$.
By the definition of $X_i(v_i)$, each $x_{v_iv_j}$ belongs to $X_i(v_i)$.
Thus, $S\cap X_i\subseteq X_i(v_i)$.
Since $W_{i,v_i}^t=\{q_{i,v_i},c_i^t\}\cup A_i\cup(X_i\setminus X_i(v_i))$
and $S\cap(Q\cup A)=\emptyset$, we have
\[W_{i,v_i}^t\cap S=\{c_i^t\}\,.\]
Moreover, for each $i<j$, the set $S$ contains only $x_{v_iv_j}$ from $X_{ij}$, and hence
\[Y_{ij}\cap S=\{x_{v_iv_j}\}\,.\]
By Lemma~\ref{lem:structure}(a), all these cliques are maximal.
Thus, every vertex of $S$ has a private clique, and hence $S$ is minimal by Lemma~\ref{lem:private}.
Furthermore,
\[|S|=|C|+\binom r2\ge |C|= r(2r^2-r+1)\,.\]

\par\medskip\noindent\textbf{Sufficiency.}
Suppose that $H$ has a minimal clique transversal $S$ with $|S|\ge r(2r^2-r+1)$.
We first show that $S$ intersects every $C_i$.
If $S\cap C_i=\emptyset$ for some $i$,
then $|S\cap C|\le(r-1)(2r^2-r+1)$.
Together with Lemma~\ref{lem:bound}, this gives
\begin{equation}\label{eq:all-groups-hit}
\begin{aligned}
|S|&=|S\cap C|+|S\setminus C|\\
&\le(r-1)(2r^2-r+1)+(2r^2-r)\\
&=r(2r^2-r+1)-1<r(2r^2-r+1),
\end{aligned}
\end{equation}
contradicting the assumption. 
Thus, for each $i$, we can choose a vertex $c_i^{t_i}\in S\cap C_i$.

Since $S$ is minimal, each $c_i^{t_i}$ has a private clique by Lemma~\ref{lem:private}.
By Lemma~\ref{lem:structure}(a), this clique is either $D_i^{t_i}$ or
$W_{i,v}^{t_i}$ for some $v\in V_i$, and both contain all of $A_i$.
By definition, a private clique of $c_i^{t_i}$ contains no other vertex of $S$,
and hence $S\cap A_i=\emptyset$.
This holds for every $i$, and thus $S\cap A=\emptyset$.

For each pair of colors $i<j$, the set $Y_{ij}=\{a_{ij},a_{ji}\}\cup X_{ij}$
is a maximal clique by Lemma~\ref{lem:structure}(a).
The set $S$ must intersect it, but $a_{ij},a_{ji}\notin S$, and hence
\[S\cap X_{ij}\ne\emptyset\,.\]
Thus, for each $i$, choosing a color $j\ne i$,
we see that $D_i^{t_i}$ contains a vertex of $S\cap X_{ij}$ in addition to $c_i^{t_i}$.
Therefore, $D_i^{t_i}$ cannot be a private clique of $c_i^{t_i}$.
Hence, there is a vertex $v_i\in V_i$ such that
\[W_{i,v_i}^{t_i}\cap S=\{c_i^{t_i}\}\,.\]
Since $W_{i,v_i}^{t_i}$ contains all of $X_i\setminus X_i(v_i)$,
this equality implies $S\cap(X_i\setminus X_i(v_i))=\emptyset$, that is,
\[S\cap X_i\subseteq X_i(v_i)\,.\]

We have now chosen $v_i$ for each color $i$.
Fix an arbitrary pair of colors $i<j$.
Since $X_{ij}\subseteq X_i$ and $S\cap X_i\subseteq X_i(v_i)$, we have
\[S\cap X_{ij}\subseteq S\cap X_i\subseteq X_i(v_i)\,.\]
Similarly, for color $j$, we have
\[S\cap X_{ij}\subseteq S\cap X_j\subseteq X_j(v_j)\,.\]
Thus, $S\cap X_{ij}$ is contained in both $X_i(v_i)$ and $X_j(v_j)$.
Moreover, we have already shown that $S\cap X_{ij}\ne\emptyset$, and hence
\[\emptyset\ne S\cap X_{ij}\subseteq X_i(v_i)\cap X_j(v_j)\,.\]
The set $X_i(v_i)$ consists of vertices representing edges incident to $v_i$,
and $X_j(v_j)$ consists of vertices representing edges incident to $v_j$.
Thus, the existence of a vertex in their intersection means that the original graph has the edge $v_iv_j$.
Since this holds for every pair of colors, $\{v_1,\ldots,v_r\}$ is a multicolored clique.
This proves both directions.
\end{proof}

\begin{proof}[Proof of Theorem~\ref{thm:main}]
Taking $(H,r(2r^2-r+1))$ as the instance of \UCT{} gives a parameterized reduction,
whose correctness follows from Lemma~\ref{lem:correctness}.
\end{proof}

\paragraph{Use of generative AI.}
The proof of Theorem~\ref{thm:main} was developed from an initial reduction and a draft proof generated through interactions with ChatGPT (OpenAI; GPT-5.6 Sol).
ChatGPT (OpenAI; GPT-6 Astra) was also used to draft the introduction in accordance with the authors' instructions and to refine the exposition.
ChatGPT assisted with the preparation of LaTeX code and TikZ figures.
The authors take full responsibility for the final manuscript.

\bibliographystyle{plainurl} 
\bibliography{biblio.bib}
\end{document}